\documentclass[letterpaper, 10 pt, conference]{ieeeconf}  

\IEEEoverridecommandlockouts                              

\usepackage{graphicx,subcaption}
\usepackage{amsmath,amssymb}
\usepackage{mathrsfs}
\usepackage{booktabs}
\usepackage{color}
\usepackage{mathrsfs}
\usepackage{algorithm}

\newtheorem{theorem}{Theorem}
\newtheorem{assumption}{Assumption}

\newtheorem{lemma}{Lemma}
\newtheorem{definition}{Definition}

\newtheorem{remark}{Remark}

\newcommand{\OneTwo}[2]
{\begin{bmatrix} {#1} & {#2}
\end{bmatrix}
}
\newcommand{\TwoTwo}[4]
{\begin{bmatrix}
{#1} & {#2} \\
{#3} & {#4}
\end{bmatrix}
}

\newcommand{\Ltwo}{\boldsymbol{\rm L}_{2}}
 
\newcommand{\Ltwoe}{\boldsymbol{\rm L}_{2e}}

\title{\LARGE \bf
Consensus in Plug-and-Play Heterogeneous Dynamical Networks: A Passivity Compensation Approach
}

\author{Yongkang Su$^{1}$, Sei Zhen Khong$^{2}$  and Lanlan Su$^{3}$
\thanks{This work was supported in part by the National Science and Technology Council of Taiwan (grant numbers: 114-2622-8-110-003, 113-2222-E-110-002-MY3 and 114-2218-E-007-011).}
\thanks{$^{1}$Y. Su is with the School of Electrical and Electronic Engineering, University of Sheffield, S10 2TN, Sheffield, UK.
        {\tt\small ysu34@sheffield.ac.uk}}%
\thanks{$^{2}$S. Z. Khong is with the Department of Electrical Engineering, National Sun Yat-sen University, Kaohsiung 804201, Taiwan.
        {\tt\small szkhong@mail.nsysu.edu.tw}}%
\thanks{$^{3}$L. Su  is with the Department of Electrical and Electronic Engineering, University of Manchester, M13 9PL, Manchester, UK.
        {\tt\small lanlan.su@manchester.ac.uk}}
}

\begin{document}

\maketitle
\thispagestyle{empty}
\pagestyle{empty}

\begin{abstract}
This paper investigates output consensus in heterogeneous dynamical networks within a plug-and-play framework. The networks are interconnected through nonlinear diffusive couplings and operate in the presence of measurement and communication noise. Focusing on systems that are input feedforward passive (IFP), we propose a passivity-compensation approach that exploits the surplus passivity of coupling links to locally offset shortages of passivity at the nodes. This mechanism enables subnetworks to be interconnected without requiring global reanalysis, thereby preserving modularity. Specifically, we derive locally verifiable interface conditions, expressed in terms of passivity indices and coupling gains, to guarantee that consensus properties of individual subnetworks are preserved when forming larger networks.
\end{abstract}


\section{Introduction}

The growing scale and complexity of modern networked systems, such as power grids, robotic swarms, transportation networks, and communication infrastructures, have motivated the development of plug-and-play frameworks \cite{crisostomi2014plug,riverso2014plug,studli2023plug}. A defining feature of such frameworks is that new components or subnetworks can be integrated into an existing network without having to redesign the entire system to satisfy certain control objectives. This modularity is crucial for ensuring scalability and flexibility in large-scale dynamical networks.

Traditionally, plug-and-play framework refers to the case where a single system (a.k.a, node) is added to or removed from a network. Each system is designed to satisfy local criteria so that the global system maintains its desired properties, such as stability or consensus, when systems join or leave \cite{riverso2014plug,sadabadi2017plug}. This concept was naturally extended to network-level plug-and-play, where an entire subnetwork (a group of interconnected systems) is treated as the unit to be integrated \cite{zhong2024secure}. The central challenge in plug-and-play is to preserve global properties of interest while requiring modifications only at the interfaces (i.e., boundary nodes or coupling links) between the new and existing network. This yields a modular design architecture: subnetworks can be interconnected while leaving their internal structures unchanged, with only interface conditions to be verified.

A fundamental property in dynamical networks is consensus, where systems’ outputs synchronise despite heterogeneity, disturbances and interconnections. Consensus is vital in applications ranging from coordination of autonomous vehicles to synchronisation in distributed energy systems \cite{jadbabaie2003coordination,olfati2007consensus,dorfler2013synchronization}. Preserving consensus in a plug-and-play setting is thus of both theoretical and practical significance, and fruitful research results have been developed, see, e.g., \cite{studli2023plug,qu2014modularized,manfredi2023distributed}. However, the plug-and-play frameworks considered in these works are limited to a single system being added to or removed from the network. To this end, we adopt a passivity-based approach in this paper to address a more general scenario of network-level plug-and-play consensus.

Passivity theory provides powerful distributed tools for analysing large-scale interconnected systems, since passivity is compositional in the sense that the passivity of subsystems implies passivity of their interconnections through appropriate mappings, such as those defined by a network structure \cite{bai2011cooperative}. In particular, we consider systems that are not necessarily passive but can be characterised by input feedforward passivity (IFP) indices, which quantify surpluses or shortages of passivity. The class of IFP systems extends beyond passive systems and, in particular, includes all asymptotically stable linear systems \cite{proskurnikov2017simple} that are not passive. Inspired by the passivity theorem, the feedback interconnection of two systems is passive if the passivity deficit of one system can be compensated by the passivity surplus of the other \cite{van2000l2}. In the network setting, passivity shortages can then be compensated locally via the coupling links, contributing to the overall network passivity \cite{su2025passivitycompensationdistributedapproach}. This passivity-compensation mechanism provides a natural basis for plug-and-play design: when subnetworks are joined, local passivity shortages at the boundary nodes can be compensated through the corresponding coupling links, and only the interface conditions need to be verified to ensure the enlarged network preserves consensus.

This paper analyses input–output (IO) consensus in heterogeneous networks with nonlinear diffusive couplings subject to measurement and communication noise. Building on a passivity-compensation-based framework, we establish locally verifiable conditions, expressed in terms of passivity indices and coupling gains, that ensure consensus properties are preserved under plug-and-play operations. In particular, we address both system-level and network-level plug-and-play scenarios, showing that consensus in the enlarged network can be guaranteed by checking only interface conditions at the boundary nodes, without the need for global reanalysis.

\section{Notation and Preliminaries}
We begin by introducing the notation. Let $\mathbb{R}$ be the set of real numbers. Given a matrix $A$, let $A^\top$ denote its transpose, and the notation $A\succ 0$ indicates that $A$ is positive definite. Let $\textbf{1}_m := [1, \ldots ,1]^{\top} \in {\mathbb{R}^m}$ and $\textbf{0}_m := [0, \ldots ,0]^{\top} \in {\mathbb{R}^m}$. Denote ${\rm col}\left( {{a_1}, \ldots ,{a_m}} \right) := {\left[ {a_1, \ldots ,a_m} \right]^{\top}}$ as the column vector with scalars ${{a_1}, \ldots ,{a_m}}$, and let $\mathrm{diag}\{v_1,\dots, v_m\}$ denote the diagonal matrix with diagonal entries $v_1, \ldots, v_m$. Denote by $\Ltwo$ the space of signals $x:\left[ {0,\infty } \right) \to {\mathbb{R}^m}$ satisfying $\int_0^\infty  {{{\left| {x(t)} \right|}^2}dt < \infty}$ with $|\cdot|$ being the Euclidean norm. Define $\Ltwoe: =\{ {x:\left[ {0,\infty } \right) \to {\mathbb{R}^m}\;|\;{P_T}x \in \Ltwo ,\forall T \ge 0} \}$, where $P_T$ is the truncation operator that satisfies $\left( {{P_T}x} \right)(t) = x(t)$ for $t \le T$ and $\left( {{P_T}x} \right)(t) = 0$ for $t>T$. For $x,y \in \Ltwoe $ and $T \ge 0$, ${\left\| x \right\|_T} := {\left( {\int_0^T {{{\left| {x(t)} \right|}^2}dt} } \right)^{\frac{1}{2}}}$ and ${\left\langle {x,y} \right\rangle _T} := \int_0^T {{x^{\top}}(t)y(t)dt}$. An operator $H:\Ltwoe \to \Ltwoe$ is said to be causal if ${P_T}H{P_T} = {P_T}H$ for all $T \ge 0$.

\vspace{2mm}
We now recall the concept of input feedforward passivity (IFP), which generalises passivity by 
quantifying surpluses or shortages of passivity. 

\begin{definition}[\cite{DesVid75}]\label{def: IFP}
A causal operator $H:\Ltwoe \to \Ltwoe $ is said to be input feedforward passive (IFP) if there exist $\nu \in \mathbb{R}$ and $\delta \in \mathbb{R}$ such that
\begin{align}\label{eq: IFP}
 {\left\langle {u,Hu } \right\rangle _T}\ge \nu\left\| {u} \right\|_T^2 + \delta,\,\forall u \in \Ltwoe,\,\forall T \ge 0.
\end{align}
\end{definition}

The positive (or negative) sign of passivity index $\nu_i$ reflects a surplus (or shortage) of passivity. 

\vspace{2mm}

The interconnection structure of a network is represented by an undirected graph $\mathcal{G} = (\mathcal{N},\mathcal{E} )$, where $\mathcal{N} = \{ 1, \ldots ,n\} $ is the node set and $\mathcal{E} \subset \mathcal{N} \times \mathcal{N}$ is the edge set. An edge $(i,j) \in \mathcal{E}$ indicates that node $i$ can receive information from node $j$. For each node $i\in\mathcal{N}$, define the set of its neighbouring nodes as $\mathcal{N}_i=\left\{ {j \in \mathcal{N}\;|\;(i,j) \in \mathcal{E}} \right\}$. The graph $\mathcal{G}$ is said to be undirected if the existence of an edge $(i, j) \in \mathcal{E}$ implies that $(j, i) \in \mathcal{E}$. An undirected graph is said to be connected if there exists a sequence of edges connecting any two nodes. The structure of a graph can be described by its incidence matrix $D=[d_{ik}]\in\mathbb{R}^{n\times p}$ with $p$ being the cardinality of $\mathcal{E}$. For an undirected graph $\mathcal{G}$, the ends of each edge $k$ are assigned a ``$+$'' and a ``$-$'' arbitrarily, and then we can denote by $\mathscr{L}_i^+$ and $\mathscr{L}_i^-$ the set of edges for which node $i$ is the positive end and negative end, respectively. The incidence matrix is then given by
\begin{equation*}
{d_{ik}} = \left\{ 
\begin{matrix}
     + 1, & k \in \mathscr{L}_i^ + \\
 - 1, & k \in \mathscr{L}_i^ - \\
0,& \mathrm{otherwise}.
\end{matrix}
\right.
\end{equation*}
By definition, $D^{\top} \textbf{1}_n =0$ when the graph $\mathcal{G}$ is undirected. In the following, we establish a lemma on matrix positive definiteness by exploiting the properties of the incidence matrix, which plays a crucial role in the subsequent analysis of output consensus.\par
\begin{lemma}\label{lem: positive definite}
Given an undirected and connected graph $\mathcal{G} = (\mathcal{N}, \mathcal{E})$, let  $D\in\mathbb{R}^{n\times p}$  be  its incident matrix,  and let  $\Theta =\mathrm{diag}\{\theta_1,\ldots,\theta_n\}$ and $\Sigma =\mathrm{diag}\{\sigma_1,\ldots,\sigma_p\}$.
It holds that  
\begin{align*}
    M := D^{\top}\Theta D+\Sigma \succ 0
\end{align*}
if there exists a diagonal matrix $S=\mathrm{diag}\left\{ {{s _1}, \ldots, s_p} \right\}\succ 0$ such that for each edge $\left( {i,j} \right)\in \mathcal{E}$, 
\begin{align}\label{eq: lemma}
{s_k}\left( {{\theta _i} + {\theta _j} + {\sigma_k}} \right) \ge \sum\limits_{{l\neq k},\,{l \in \mathscr{L}_i^ \pm }} {{s_l}\left| {{\theta _i}} \right|}  + \sum\limits_{{l\neq k},\,{l \in \mathscr{L}_j^ \pm }} {{s_l}\left| {{\theta _j}} \right|},
\end{align}
where $k \in \mathscr{L}_i^ + \cap \mathscr{L}_j^ -$ and $\mathscr{L}_i^ \pm:=\mathscr{L}_i^ + \cup \mathscr{L}_i^-$.\par
The proof of Lemma~\ref{lem: positive definite} is provided in Appendix.
\end{lemma}

\section{Problem description}
The problem studied in this work is set within a plug-and-play framework. Our objective is to establish local conditions that guarantee consensus is preserved when one network is interconnected with another. To formalize this setting, we first introduce the notion of consensus for a fixed undirected network.

Consider a network of $n$ systems described by a graph $\mathcal{G} = (\mathcal{N},\mathcal{E} )$, where each node $i\in\mathcal{N}$ corresponds to a system $H_i:\Ltwoe \to\Ltwoe$ defined by
\begin{equation}\label{eq: system model}
    {y_i} = {H_i}{u_i}, \,i\in\{1,2,\ldots,n\},
\end{equation}
where $u_i, y_i \in \Ltwoe$ are the input and output of the $i$-th system, respectively. In this paper, we consider an undirected and connected graph $\mathcal{G}$, where the systems are diffusively coupled over static nonlinear operators ${\phi_{ij}}: \Ltwoe \to \Ltwoe$, $\left( {i,j} \right)\in \mathcal{E}$ that map $0$ to $0$, and satisfy the sector-boundness: $0 < \underline{\alpha}_{ij}  \le\frac{\left( {{\phi_{ij}}\left( x \right)} \right)\left( t \right)}{{x\left( t \right)}} \le \overline{\alpha}_{ij} < \infty$ for all $x(t)\neq 0$ and  odd symmetry $\left( {{\phi _{ji}}\left( -x \right)} \right)\left( t \right) =  - \left( {{\phi _{ij}}\left( x \right)} \right)\left( t \right)$. To be specific, the input $u_i, i\in \mathcal{N}$ is expressed as
\begin{equation}\label{eq: input} 
u_i = -\sum\limits_{j \in \mathcal{N}_i} {{\phi_{ij}}\left( {y_i + w_i - y_j - w_j} \right)}, 
\end{equation}
where $w_i \in \Ltwoe, i \in \{1,2,\ldots,n\}$ are external signals that can represent measurement noise in the $i$-th system,  or $(w_i - w_j) \in \Ltwoe$ can model the communication noise present in the link between the $i$-th and $j$-th systems.

Let $U := {\rm col}\left( {{u_1}, \ldots ,{u_n}} \right)$, $Y := {\rm col}\left( {{y_1}, \ldots ,{y_n}} \right)$ and $W := {\rm col}\left( {{w_1}, \ldots ,{w_n}} \right)$. It follows from the definition of $D$ that \eqref{eq: input} can be rewritten into 
\begin{equation}\label{eq: iuput vector}
U = - D\Phi\left(D^{\top} \left( {Y + W} \right) \right), 
\end{equation}
where $\Phi : \Ltwoe \to \Ltwoe$ is defined as $\Phi\left(X\right):= {\rm col}\left( {{\phi _1}\left( {{x_1}} \right), \ldots ,{\phi _p}\left( {{x_p}} \right)} \right)$ with $X := {\rm col}\left( {{x_1}, \ldots ,{x_p}} \right)\in \Ltwoe$. Here, each component $\phi_k$ corresponds to the nonlinear map associated with the $k$-th edge, with ${\phi _k}\left(  \cdot  \right) = {\phi _{ij}}\left(  \cdot  \right)$ if $d_{ik}=1$ and $d_{jk}=-1$. The block diagram illustrating the heterogeneous network described by \eqref{eq: system model} and \eqref{eq: input} is shown in Fig. \ref{fig.network}. 
\begin{figure}
\centering
\includegraphics[width=7cm]{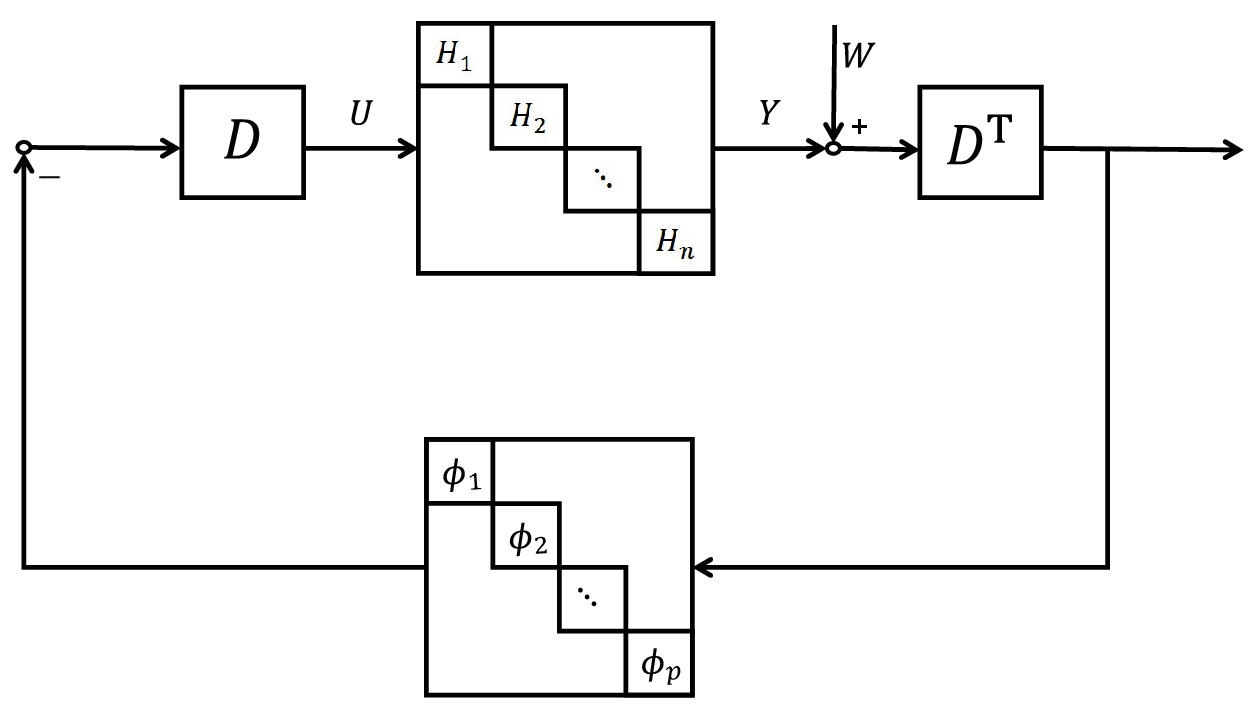}
\caption{Block diagram of the network \eqref{eq: system model} and \eqref{eq: input}.}
\label{fig.network}
\end{figure}
\vspace{2mm}
\begin{definition}\label{def: IO consensus}
The network described by \eqref{eq: system model} and \eqref{eq: input} is said to achieve input-output consensus if there exist a finite gain $\rho>0$ and a constant $\sigma\ge 0$ such that
\begin{align}\label{eq:aim}
    {\left\| {D^{\top} Y} \right\|_T}\le \rho{\left\| {D^{\top} W} \right\|_T}+\sigma,\,\forall {W } \in \Ltwoe,\,\forall T \ge 0.
\end{align}
\end{definition}

\vspace{2mm}

The term $\sigma$ in \eqref{eq:aim} allows the transient behaviour that may occur due to initial mismatches among the systems’ outputs prior to reaching consensus. Recalling the definition of incidence matrix $D$, ${\left\| {D^{\top} Y} \right\|_T}$ quantifies the mismatches among the outputs of systems in the time interval $[0,T]$, and can thus be interpreted as a measure of output synchrony on $[0,T]$. Consequently, condition \eqref{eq:aim} ensures that if the external inputs are closely aligned, the corresponding outputs will exhibit a scaled but similar level of consensus.

\vspace{2mm}

We consider a network of heterogeneous  systems \eqref{eq: system model}, each characterisable by an input passivity index and interconnected through nonlinear diffusive couplings \eqref{eq: input}. Our objective is to analyse their output consensus behaviour within a plug-and-play framework. Specifically, given two disjoint dynamical networks, we aim to derive \emph{locally verifiable interface conditions} under which the interconnected larger network preserves the consensus properties of the individual subnetworks.
To this end, we employ a passivity-compensation approach: any passivity shortage in a system can be \emph{locally} compensated through its coupling links, which contribute to the overall network passivity and hence ensure consensus. This local compensation mechanism can then be systematically exploited in the plug-and-play interface design to guarantee that consensus is maintained when subnetworks are interconnected.


\section{Main Results}

In Section \ref{sec: Standard plug-and-play consensus}, we develop the main result  in the standard plug-and play setting in which a single system joins a network. The  analysis is extended  in Section \ref{sec: Network Plug-and-play consensus} to the more general case where a subnetwork is plugged into an existing network.

\subsection{Single-node plug-and-play consensus}\label{sec: Standard plug-and-play consensus}
Consider the scenario where a new system $H_{n+1}$ is added to the network described by \eqref{eq: system model} and  \eqref{eq: input}. For simplicity, we assume $H_{n+1}$ is  connected to the existing graph $\mathcal{G}$ through a single edge. 
Let $\mathcal{G}_{new}$ denotes the augmented graph obtained when the new system $H_{n+1}$ is connected to $\mathcal{G}$. Compared to $\mathcal{G}$, $\mathcal{G}_{new}$ contains an additional node and an additional edge. 

For each node $i\in\mathcal{N}$, let $r_i$ denote the number of neighbours of system $H_i$. The following theorem proposes local conditions on the added node and its connecting edge that ensure IO consensus of the network is preserved after the ``plug and play".  

\vspace{2mm}

\begin{theorem}\label{thm: add a system}
Consider a new system $H_{n+1}$ joining the network described by \eqref{eq: system model} and \eqref{eq: input} by establishing a connection with some system $H_c,\, c\in\{1,\ldots,n\}$, through a sector-bounded coupling operator ${\phi_{(n+1)c}}(\cdot)$. Suppose that all systems in the original network are IFP with indices $\nu_i,i\in\{1,\ldots,n\}$, and that the following condition is satisfied for each edge $(i,j)\in\mathcal{E}$ in $\mathcal{G}$
\begin{align}\label{eq: G}
\frac{1}{{{{\overline{\alpha} }_{ij}}}} + {\nu _i} + {\nu _j} - \left( {{r_i} - 1} \right)\left| {{\nu _i}} \right| - \left( {{r_j} - 1} \right)\left| {{\nu _j}} \right| > 0. 
\end{align}
Then, the original network described by \eqref{eq: system model} and \eqref{eq: input} achieves IO consensus. Moreover, under condition \eqref{eq: G},  IO consensus is preserved in the augmented network  if $H_{n+1}$ is IFP with index $\nu_{n+1}$, and
\begin{align}\label{eq: n+1}
    \gamma\left( {\frac{1}{{{{\overline{\alpha} }_{(n + 1)c}}}} + {\nu _{n + 1}} + {\nu _c}} \right) - r_c\left| {{\nu _c}} \right| > 0,
\end{align}
where
\begin{align*}
    \gamma:=\min\limits_{j \in {\mathcal{N}_c}} \frac{{\frac{1}{{{{\overline{\alpha} }_{cj}}}} + {\nu _c} + {\nu _j} - \left( {{r_c} - 1} \right)\left| {{\nu _c}} \right| - \left( {{r_j} - 1} \right)\left| {{\nu _j}} \right|}}{{\left| {{\nu _c}} \right|}},
\end{align*}
and ${{{\overline{\alpha} }_{(n+1)c}}}$ is the upper sector bound of ${\phi _{(n+1)c}}(\cdot)$.
\end{theorem}

\begin{proof}
First, we show that if \eqref{eq: G} holds, the IO consensus is achieved in the original network. To this end, let $V:=\Phi\left(D^{\top} \left( {Y + W} \right) \right) $. For all $U\in \Ltwoe$, it follow from Definition \ref{def: IFP} that
\begin{align}\label{eq: U,Y}
{\left\langle {V, - {D^{\top}}Y} \right\rangle _T}&={\left\langle {-DV, Y} \right\rangle _T} = {\left\langle {U,Y} \right\rangle _T} \nonumber\\
&\ge {\nu _1}\left\| {{u_1}} \right\|_T^2 + \delta_1 +  \cdots {\nu _n}\left\| {{u_n}} \right\|_T^2 +\delta_n\nonumber\\
&= {\left\langle {U,\Psi U} \right\rangle _T}+\bar \delta\nonumber\\
&={\left\langle {V, {{D^{\top}}\Psi D}V} \right\rangle _T}+\bar \delta.
\end{align}
where $\Psi : =\mathrm{diag}\left\{ {{\nu _1}, \dots, {\nu _n}} \right\}$ and $\bar \delta : = \sum\limits_{i = 1}^n {{\delta _i}} $. On the other hand, one has
\begin{align}\label{eq: V,Y}
& \, {\left\langle {V,{D^{\top}}\left( {Y + W} \right)} \right\rangle _T} = {\left\langle {\Phi\left(D^{\top} \left( {Y + W} \right) \right),{D^{\top}}\left( {Y + W} \right)} \right\rangle _T} \nonumber\\
&\ge {\left\langle {\Phi\left(D^{\top} \left( {Y + W} \right) \right),\Lambda \Phi\left(D^{\top} \left( {Y + W} \right) \right)} \right\rangle _T}= {\left\langle {V,\Lambda V} \right\rangle _T},
\end{align}
where $\Lambda:=\mathrm{diag}\{ {\alpha _1}, \ldots ,{\alpha _p}\}$ with ${\alpha _k} = \frac{1}{{{{\overline{\alpha} }_{ij}}}}$ if $d_{ik}=1$ and $d_{jk}=-1$, and the inequality holds due to  ${\left\langle {x,{\phi_{ij}}\left( x \right)} \right\rangle _T}\ge\frac{1}{{{\overline{\alpha}_{ij}}}}\left\| {{\phi _{ij}}\left( x \right)} \right\|_T^2$. It follows from \eqref{eq: U,Y} and \eqref{eq: V,Y} that
\begin{align}\label{eq: V,W}
{\left\langle {V,{D^{\top}}W} \right\rangle _T} &= {\left\langle {V, - {D^{\top}}Y} \right\rangle _T} + {\left\langle {V,{D^{\top}}\left( {Y + W} \right)} \right\rangle _T}\nonumber\\
&\ge {\left\langle {V,\left( {{D^{\top}}\Psi D + \Lambda } \right)V} \right\rangle _T}+\bar \delta.
\end{align}
By hypothesis, $\frac{1}{{{{\overline{\alpha} }_{ij}}}} + {\nu _i} + {\nu _j} - \left( {{r_i} - 1} \right)\left| {{\nu _i}} \right| - \left( {{r_j} - 1} \right)\left| {{\nu _j}} \right| > 0$ for all $ (i,j)\in\mathcal{E}$. According to Lemma \ref{lem: positive definite} and let $s_k = 1, k \in \{1,\dots,p\}$, ${D^{\top}}\Psi D + \Lambda \succ 0$. Let $\kappa >0$ be the smallest eigenvalue of ${D^{\top}}\Psi D + \Lambda$. By \eqref{eq: V,W} as well as the fact that $ab \le \frac{{{\kappa}}}{2}{a^2} +\frac{1}{{2{\kappa}}}{b^2}$, we have
\begin{align*}
\kappa \left\| V \right\|_T^2   
& \le {\left\langle {V,{D^{\top}}W} \right\rangle _T} - \bar \delta\\
& \le \frac{\kappa }{2}\left\| V \right\|_T^2 + \frac{1}{{2\kappa}}\left\| {{D^{\top}}W} \right\|_T^2 - \bar \delta,
\end{align*}
which yields that
$$\left\| V \right\|_T^2 \le \frac{1}{{{\kappa^2}}}\left\| {{D^{\top}}W} \right\|_T^2 - \frac{{2\bar \delta }}{\kappa }.$$
It follows from ${a^2} \pm {b^2} \le  {\left( {\left| a \right| + \left| b \right|} \right)^2}$ that ${\left\| V \right\|_T} \le \frac{1}{\kappa}{\left\| {{D^{\top}}W} \right\|_T} + \sqrt {\frac{{2\left| {\bar \delta } \right|}}{\kappa }}.$ Since ${\left\| {V} \right\|_T} \ge \underline{\alpha} {\left\| {{D^{\top}}\left( {Y + W} \right)} \right\|_T}$ with $\underline{\alpha} := \mathop {\min }\limits_{(i,j) \in \mathcal{E}} {\underline{\alpha}_{ij}}$ and ${\underline{\alpha}_{ij}}$ is the lower sector bound of $\phi_{ij}(\cdot)$, it can be obtained from the fact $\left| {a + b} \right| \ge \left| a \right| - \left| b \right|$ that
$${\left\| {{D^{\top}}Y} \right\|_T} \le \left( {\frac{1}{{\kappa \underline{\alpha} }} + 1} \right){\left\| {{D^{\top}}W} \right\|_T} + \frac{1}{{\underline{\alpha} }} \sqrt {\frac{{2\left| {\bar \delta } \right|}}{\kappa }}.$$
Noting that the above reasoning holds for all $T>0$,  the original network achieves IO consensus.\par

Next, we examine the IO consensus of the augmented network that results from connecting the additional system $H_{n+1}$ to the original network. For the augmented graph $\mathcal{G}_{\mathrm{new}}$, denote by $\mathscr{\bar L}_i^+$ and $\mathscr{\bar L}_i^-$ the set of edges for which node $i$ is the positive end and negative end, respectively, and let $\mathscr{\bar L}_i^ \pm:=\mathscr{\bar L}_i^+ \cup \mathscr{\bar L}_i^-$. Let $D_{\mathrm{new}} = [d_{ik}^{\mathrm{new}}] \in \mathbb{R}^{(n+1) \times (p+1)}$ be the incidence matrix associated with $\mathcal{G}_{\mathrm{new}}$. The incidence matrix $D_{\mathrm{new}}$ can be set as $D_{\mathrm{new}} = \TwoTwo{D}{\Upsilon}{\textbf{0}_n^\top}{1}$, where $\Upsilon : = [\upsilon_i] \in \mathbb{R}^{n \times 1}$ with $\upsilon_i = -1$ if $i=c$, otherwise $\upsilon_i = 0$. Let $\Psi_{\mathrm{new}}  := \mathrm{diag}\left\{ {\Psi, {\nu _{n+1}}} \right\}$ with $\Psi  =\mathrm{diag}\left\{ {{\nu _1}, \dots, {\nu _n}} \right\}$ and $\Lambda_{\mathrm{new}} : = \mathrm{diag}\left\{ {\Lambda, \frac{1}{{{{\overline{\alpha} }_{(n + 1)c}}}}} \right\}$ with $\Lambda=\mathrm{diag}\{ {\alpha _1}, \ldots ,{\alpha _p}\}$. 

To determine the positive definiteness of matrix $D_{\mathrm{new}}^{\top}\Psi_{\mathrm{new}}D_{\mathrm{new}}+\Lambda_{\mathrm{new}}\in\mathbb{R}^{(p+1) \times (p+1)}$, we first set $S = \mathrm{diag}\{s_1,\dots,s_{p+1}\}$ with $s_k = 1, k\in\{1,\dots,p\}$ and $s_{p+1} = \gamma > 0$. Noting that $p+1\in \mathscr{\bar L}_{n+1}^ + \cap \mathscr{\bar L}_c^ -$,
\begin{align*}
&\sum\limits_{l\neq p+1,\,l \in \mathscr{\bar L}_c^ \pm } {{s_l}\left| {{\nu _c}} \right|} = r_c\left| {{\nu _c}} \right|,\,\sum\limits_{l\neq p+1,\,l \in \mathscr{\bar L}_{n+1}^ \pm } {{s_l}\left| {{\nu _{n+1}}} \right|} = 0,\\
&\sum\limits_{l\neq k,\,l \in \mathscr{\bar L}_j^ \pm } {{s_l}\left| {{\nu _j}} \right|} = (r_j-1)\left| {{\nu _j}} \right|,\, k \in \mathscr{\bar L}_j^ \pm,\forall j\neq c,n+1,\\
&\sum\limits_{l\neq k,\,l \in \mathscr{\bar L}_c^ \pm } {{s_l}\left| {{\nu _c}} \right|} = (r_c-1+\gamma)\left| {{\nu _c}} \right|,\, k \in \mathscr{\bar L}_c^ \pm ,k\neq p+1.
\end{align*}
It follows from Lemma \ref{lem: positive definite} that $D_{\mathrm{new}}^{\top}\Psi_{\mathrm{new}}D_{\mathrm{new}}+\Lambda_{\mathrm{new}}$ is positive definite if the following conditions hold: 1) $\gamma\left( {\frac{1}{{{{\overline{\alpha} }_{(n + 1)c}}}} + {\nu _{n + 1}} + {\nu _c}} \right) > r_c\left| {{\nu _c}} \right|$; 2) $\frac{1}{{{{\overline{\alpha} }_{cj}}}} + {\nu _c} + {\nu _j} > \left( {{r_c} - 1} \right)\left| {{\nu _c}} \right| + \left( {{r_j} - 1} \right)\left| {{\nu _j}} \right|+\gamma\left| {{\nu _c}} \right|$ for all $(c,j)\in\mathcal{E}$; 3) $\frac{1}{{{{\overline{\alpha} }_{ij}}}} + {\nu _i} + {\nu _j} > \left( {{r_i} - 1} \right)\left| {{\nu _i}} \right| + \left( {{r_j} - 1} \right)\left| {{\nu _j}} \right|$ for all $(i,j)\in\mathcal{E}$ and $i,j \neq c$. Consequently, if \eqref{eq: G} and \eqref{eq: n+1} hold, the aforementioned three conditions are simultaneously satisfied, leading to the conclusion that $D_{\mathrm{new}}^{\top}\Psi_{\mathrm{new}}D_{\mathrm{new}}+\Lambda_{\mathrm{new}}$ is positive definite. Following the same reasoning line with the IO consensus analysis for the original network, it can be obtained that there exist a finite gain $\rho>0$ and a constant $\sigma\ge 0$ such that
\begin{align*}
    {\left\| {D_{\mathrm{new}}^{\top} Y_{\mathrm{new}}} \right\|_T}\le \rho{\left\| {D_{\mathrm{new}}^{\top} W_{\mathrm{new}}} \right\|_T}+\sigma,\,\forall T \ge 0,
\end{align*}
where $Y_{\mathrm{new}}: = \left[ {\begin{array}{*{20}{c}}
Y\\
{{y_{n + 1}}}
\end{array}} \right]$ and $W_{\mathrm{new}} :=  \left[ {\begin{array}{*{20}{c}}
W\\
{{w_{n + 1}}}
\end{array}} \right]$. 
\end{proof}

\vspace{2mm}

\begin{remark}
Theorem~\ref{thm: add a system} provides a practical criterion for preserving IO consensus in a plug-and-play setting. Specifically, when a new IFP system is connected to an existing network in which all edges satisfy the distributed condition~\eqref{eq: G}, it is sufficient to select the coupling at the interface such that condition~\eqref{eq: n+1} holds. In practice, this requires the node $H_c$, to which the new system is connected, to have knowledge of its own passivity index, together with the passivity indices, the number of neighbours, and the associated coupling functions of its neighbours.
\end{remark}

\subsection{Network plug-and-play consensus}\label{sec: Network Plug-and-play consensus}

Next, we extend the plug-and-play framework from the system level to the network level. Specifically, we consider the case where an entire subnetwork is interconnected with an existing network via a set of designated boundary links.

Given two separated undirected and connected graphs $\mathcal{G}_1=\left( {\mathcal{N}_1,\mathcal{E}_1} \right)$ with $\mathcal{N}_1 = \{ 1, \ldots ,n\} $ and $\mathcal{G}_2=\left( {\mathcal{N}_2,\mathcal{E}_2} \right)$ with $\mathcal{N}_2 = \{ n+1, \ldots ,m\} $. Let $\mathcal{G}_1$ and $\mathcal{G}_2$ be networks defined in the same manner as network described by \eqref{eq: system model} and \eqref{eq: input}. Suppose the two disjoint graphs $\mathcal{G}_1,\mathcal{G}_2$ are connected by a set of new edges  $\mathcal{E}_{12}$. As such, each new edge $(i,j)\in \mathcal{E}_{12}$ connects a node   $i\in\mathcal{N}_1$ to another node $j\in\mathcal{N}_2$. The graph obtained by interconnecting $\mathcal{G}_1$ and $\mathcal{G}_2$ is ${\mathcal{G}_{12}} = \left(\mathcal{N}_1\cup\mathcal{N}_2,\mathcal{E}_1\cup\mathcal{E}_2\cup{\mathcal{E}_{12}}\right)$.



It is assumed that the set of boundary edges $\mathcal{E}_{12}$ is selected according to the following condition.  

\begin{assumption}\label{assum: network to network}
For every pair of distinct edges $\left(i_r,j_r\right),\left(i_s,j_s\right)\in\mathcal{E}_{12}$, the boundary nodes $i_r,i_s \in \mathcal{N}_1$ are non-adjacent in $\mathcal{G}_1$, and the boundary nodes $j_r,j_s \in \mathcal{N}_2$ are non-adjacent in $\mathcal{G}_2$. That is, $(i_r,i_s)\notin\mathcal{E}_1$ and $(j_r,j_s)\notin\mathcal{E}_2$ for all $r\neq s$.  
\end{assumption}

\begin{figure}[htbp]
\centering
\subfloat[]
{
\label{fig:subfig1}\includegraphics[width=0.2\textwidth]{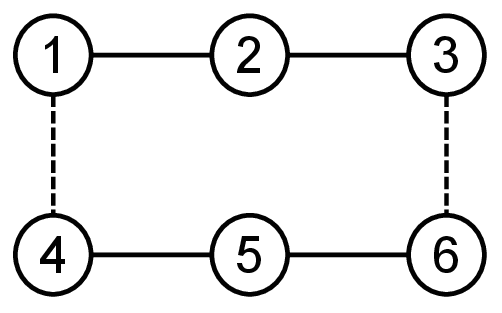}
}
  \subfloat[]
  {      \label{fig:subfig2}\includegraphics[width=0.2\textwidth]{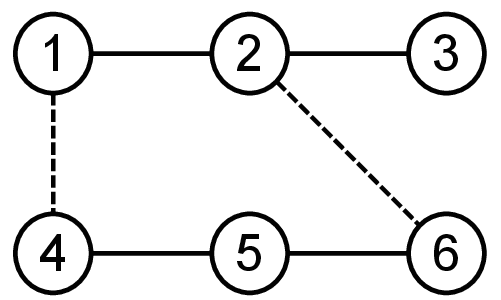}
  }
  \caption{Illustration of the interconnection assumption between two connected graphs 
    $\mathcal{G}_1$ and $\mathcal{G}_2$.}
  \label{fig: assuption}
\end{figure}

To illustrate Assumption~\ref{assum: network to network}, consider two undirected and connected graphs, $\mathcal{G}_1$ and $\mathcal{G}_2$, for example. Let $\mathcal{G}_1$ have node set $\mathcal{N}_1 = \{1,2,3\}$ and edge set $\mathcal{E}_1 = \{(1,2),(2,3)\}$, and let $\mathcal{G}_2$ have node set $\mathcal{N}_2 = \{4,5,6\}$ and edge set $\mathcal{E}_2 = \{(4,5),(5,6)\}$.  
If the interconnecting edges are chosen as $\mathcal{E}_{12} = \{(1,4),(3,6)\}$, as shown in Fig.~\ref{fig: assuption}(a), then the boundary nodes $i_1=1$ and $i_2=3$ in $\mathcal{G}_1$ are non-adjacent, and the boundary nodes $j_1=4$ and $j_2=6$ in $\mathcal{G}_2$ are also non-adjacent. Hence, the assumption is satisfied.   By contrast, if $\mathcal{E}_{12} = \{(1,4),(2,6)\}$, as shown in Fig.~\ref{fig: assuption}(b), then the boundary nodes $i_1=1$ and $i_2=2$ in $\mathcal{G}_1$ are adjacent, which violates the assumption.


For each node $i\in\mathcal{N}_1\cup\mathcal{N}_2$, let $r_i$ denote the number of neighbours of system $H_i$. The following theorem establishes local conditions on the added subnetwork and its interconnecting edges that ensure the IO consensus of the overall network is preserved after the network plug-and-play.

\vspace{2mm}

\begin{theorem}\label{thm: network to network}
Consider two subnetworks represented by graphs $\mathcal{G}_1$ and $\mathcal{G}_2$ interconnected via a set of coupling links $\mathcal{E}_{12}$ that satisfies Assumption \ref{assum: network to network}, where each link $\left( {p,q} \right)\in \mathcal{E}_{12}$ connects node $p \in \mathcal{G}_1$ and node $q \in \mathcal{G}_2$ through a sector-bounded operator ${\phi_{pq}}(\cdot)$. Suppose that all systems in the graphs $\mathcal{G}_1$ and $\mathcal{G}_2$ are IFP with indices $\nu_i,i\in\{1,\ldots,m\}$ and that the following condition is satisfied for each $(i,j)\in\mathcal{E}_1$ in $\mathcal{G}_1$ and $(i,j)\in\mathcal{E}_2$ in $\mathcal{G}_2$
\begin{align}\label{eq: two graph}
&\frac{1}{{{{{\overline\alpha} }_{ij}}}} + {\nu _i} + {\nu _j} - \left( {{r_i} - 1} \right)\left| {{\nu _i}} \right| - \left( {{r_j} - 1} \right)\left| {{\nu _j}} \right| > 0.
\end{align}
Then, both $\mathcal{G}_1$ and $\mathcal{G}_2$ achieve IO consensus. Moreover, under condition \eqref{eq: two graph}, IO consensus is preserved in the aggregate network if 
\begin{align}\label{eq: add a network}
    \gamma_{pq}\left( {\frac{1}{{{{{\overline\alpha} }_{pq}}}} + {\nu _p} + {\nu _q}} \right) - r_p\left| {{\nu _p}} \right|-r_q\left| {{\nu _q}} \right| > 0,\forall \left( {p,q} \right)\in \mathcal{E}_{12}
\end{align}
where ${{{{\overline\alpha} }_{pq}}}$ is the upper sector bound of ${\phi_{pq}}(\cdot)$ and $\gamma_{pq} = \min\{\gamma_p,\gamma_q\}$ with
\begin{align*}
    &\gamma_p:=\min\limits_{(p,j)\in\mathcal{E}_1} \frac{{\frac{1}{{{{{\overline\alpha} }_{pj}}}} + {\nu _p} + {\nu _j} - \left( {{r_p} - 1} \right)\left| {{\nu _p}} \right| - \left( {{r_j} - 1} \right)\left| {{\nu _j}} \right|}}{{\left| {{\nu _p}} \right|}},\\
    &\gamma_q:=\min\limits_{(q,j)\in\mathcal{E}_2} \frac{{\frac{1}{{{{{\overline\alpha} }_{qj}}}} + {\nu _q} + {\nu _j} - \left( {{r_q} - 1} \right)\left| {{\nu _q}} \right| - \left( {{r_j} - 1} \right)\left| {{\nu _j}} \right|}}{{\left| {{\nu _q}} \right|}}.
\end{align*}
\end{theorem}
\vspace{2mm}
\begin{proof}
By Theorem \ref{thm: add a system},  it follows directly that the networks $\mathcal{G}_1$ and $\mathcal{G}_2$ will achieve IO consensus, provided that \eqref{eq: two graph} holds for all $(i,j)\in\mathcal{E}_1$ in $\mathcal{G}_1$ and $(i,j)\in\mathcal{E}_2$ in $\mathcal{G}_2$. For $\mathcal{G}_{12}$, denote by $\mathscr{\hat L}_i^+$ and $\mathscr{\hat L}_i^-$ the sets of edges for which node $i\in\{1,\dots,m\}$ is the positive end and negative end, respectively, and let $\mathscr{\hat L}_i^ \pm:=\mathscr{\hat L}_i^+ \cup \mathscr{\hat L}_i^-$. For each $(p,q) \in \mathcal{E}_{12}$ with $p \in \mathcal{G}_1$ and $q \in \mathcal{G}_2$, we assign without loss of generality node $p$ as the positive end and node $q$ as the negative end of the edge. Let $\bar p_{12}$ and $\hat p_{12}$ be the cardinalities of $\mathcal{E}_1 \cup \mathcal{E}_2$ and $\mathcal{E}_{12}$ respectively. Let $D_1$, $D_2$ and $D_{12}$ be the incidence matrix associated with graphs $\mathcal{G}_1$, $\mathcal{G}_2$ and $\mathcal{G}_{12}$, respectively. The incidence matrix $D_{12} = [d^{12}_{ik}]\in\mathbb{R}^{(n+m)\times(\bar p_{12}+\hat p_{12})}$ is given by $D_{12} = \OneTwo{\bar D_{12}}{\hat D_{12}}$ where $\bar D_{12} = \mathrm{diag}\{D_1,D_2\}$ and $\hat D_{12}\in\mathbb{R}^{(n+m)\times\hat p_{12}}$ denotes the incidence matrix corresponding to the edge set $\mathcal{E}_{12}$. Let $\Psi_{12} := \mathrm{diag}\{\nu_1, \dots, \nu_m\}$ and $\Lambda_{12}:=\mathrm{diag}\{ {\alpha _1}, \ldots ,{\alpha _{(\bar p_{12}+\hat p_{12})}}\}$ with ${\alpha _k} = \frac{1}{{{{\overline{\alpha} }_{ij}}}}$ if $d^{12}_{ik}=1$ and $d^{12}_{jk}=-1$.\par

To determine the positive definiteness of matrix $D_{12}^{\top}\Psi_{12}D_{12}+\Lambda_{12}\in\mathbb{R}^{(\bar p_{12}+\hat p_{12}) \times (\bar p_{12}+\hat p_{12})}$, we first set $S = \mathrm{diag}\{\bar S,\hat S\}$ where $\bar S = \mathrm{diag}\{s_1,\dots, s_{\bar p_{12}}\}$ with $s_k = 1, k\in\{1,\dots,\bar p_{12}\}$ and $\hat S = \mathrm{diag}\{s_{\bar p_{12}+1},\dots,s_{\bar p_{12}+\hat p_{12}}\}$ with $s_k = \gamma_{pq}>0$ if $(p,q) \in \mathcal{E}_{12}$ and $k\in \mathscr{\hat L}^+_p\cap\mathscr{\hat L}^-_q$. For all $j\in\mathcal{N}_1\cup\mathcal{N}_2$ and $(p,q) \in \mathcal{E}_{12}$, according to Assumption~\ref{assum: network to network}, we have $\sum\limits_{l\neq k,\,l \in \mathscr{\hat L}_p^ \pm } {{s_l}\left| {{\nu _p}} \right|} = r_p\left| {{\nu _p}} \right|,\sum\limits_{l\neq k,\, l \in \mathscr{\hat L}_q^ \pm} {{s_l}\left| {{\nu _q}} \right|} = r_q\left| {{\nu _q}} \right|$ if $k\in \mathscr{\hat L}^+_p\cap\mathscr{\hat L}^-_q$, $\sum\limits_{l\neq k,\, l \in \mathscr{\hat L}_j^ \pm } {{s_l}\left| {{\nu _j}} \right|} = (r_j-1)\left| {{\nu _j}} \right|$ if $ k \in \mathscr{\hat L}_j^ \pm$ and $j\notin \{p,q\}$, and $\sum\limits_{l\neq k,\,l \in \mathscr{\hat L}_j^ \pm } {{s_l}\left| {{\nu _j}} \right|} = (r_j-1+\gamma_{pq})\left| {{\nu _j}} \right|$ if $ k \in \mathscr{\hat L}_j^ \pm$ but $k\notin \mathscr{\hat L}^+_p\cap\mathscr{\hat L}^-_q$ for $j\in\{ p,q\}$.
Since the set of coupling links $\mathcal{E}_{12}$ satisfies Assumption~\ref{assum: network to network}, the edges of $\mathcal{G}_{12}$ can be classified into four cases: 1) $(p,q) \in \mathcal{E}_{12}$; 2) $(p,j) \in \mathcal{E}_{1}$ with node $p$ is connected to $\mathcal{G}_2$; 3) $(q,j) \in \mathcal{E}_{2}$ with node $q$ is connected to $\mathcal{G}_1$; 4) $(i,j)\in\mathcal{E}_{1}\cup\mathcal{E}_{2}$ with neither node $i$ nor node $j$ is connected to another graph. Based on these four types of edges, it follows from Lemma \ref{lem: positive definite} that $D_{12}^{\top}\Psi_{12}D_{12}+\Lambda_{12}$ is positive definite if the following conditions hold: 1) $\gamma_{pq}\left( {\frac{1}{{{{\overline{\alpha} }_{pq}}}} + {\nu _p} + {\nu _q}} \right) > r_p\left| {{\nu _p}} \right|+r_q\left| {{\nu _q}} \right|$ for all $(p,q) \in \mathcal{E}_{12}$; 2) $\frac{1}{{{{\overline{\alpha} }_{pj}}}} + {\nu _p} + {\nu _j} > \left( {{r_p} - 1+\gamma_{pq}} \right)\left| {{\nu _p}} \right| + \left( {{r_j} - 1} \right)\left| {{\nu _j}} \right|$ for all $(p,j)\in\mathcal{E}_1$ and $(p,q) \in \mathcal{E}_{12}$; 3) $\frac{1}{{{{\overline{\alpha} }_{qj}}}} + {\nu _q} + {\nu _j} > \left( {{r_q} - 1+\gamma_{pq}} \right)\left| {{\nu _q}} \right| + \left( {{r_j} - 1} \right)\left| {{\nu _j}} \right|$ for all $(q,j)\in\mathcal{E}_2$ and $(p,q) \in \mathcal{E}_{12}$; 4) $\frac{1}{{{{\overline{\alpha} }_{ij}}}} + {\nu _i} + {\nu _j} > \left( {{r_i} - 1} \right)\left| {{\nu _i}} \right| + \left( {{r_j} - 1} \right)\left| {{\nu _j}} \right|$ for all $(i,j)\in\mathcal{E}_{1}\cup\mathcal{E}_{2}$ with neither node $i$ nor node $j$ is connected to another graph. Consequently, if \eqref{eq: two graph} and \eqref{eq: add a network} hold, the aforementioned four conditions are simultaneously satisfied, leading to the conclusion that $D_{12}^{\top}\Psi_{12}D_{12}+\Lambda_{12}$ is positive definite. Following the same reasoning line with the IO consensus analysis in Theorem~\ref{thm: add a system}, it can be obtained that there exist a finite gain $\rho>0$ and a constant $\sigma\ge 0$ such that
\begin{align*}
    {\left\| {D_{12}^{\top} Y_{12}} \right\|_T}\le \rho{\left\| {D_{12}^{\top} W_{12}} \right\|_T}+\sigma,\,\forall T \ge 0,
\end{align*}
where $Y_{12} := {\rm col}\left( {{y_1}, \ldots ,{y_m}} \right)$, $W_{12} := {\rm col}\left( {{w_1}, \ldots ,{w_m}} \right)$.
\end{proof}

\section{Numerical Example}

In this example, as shown in Fig. \ref{fig: add a network}, we consider two networks $\mathcal{G}_1$ consisting of systems $H_1,\dots,H_4$ and $\mathcal{G}_2$ consisting of systems $H_5,\dots,H_7$. The dynamics of the systems are given by
\begin{align*}
&{H_1} = \frac{{s + 1}}{{s(s+0.7)}}, {H_2} =  \frac{{s + 0.9}}{{s(s+0.65)}},{H_3} =\frac{{s + 0.5}}{{s(s+0.4)}},\\
&{H_4} =\frac{{\left( {s + 1.5} \right)\left( {s + 2} \right)}}{{s\left( {s + 1} \right)\left( {s + 1.8} \right)}},{H_5} = \frac{{\left( {s + 0.5} \right)\left( {s + 0.7} \right)}}{{s\left( {s + 0.45} \right)\left( {s + 0.65} \right)}},\\
&{H_6} = \frac{{\left( {s + 1} \right)\left( {s + 1.4} \right)}}{{s\left( {s + 0.8} \right)\left( {s + 1.2} \right)}},{H_7} =  \frac{{\left( {s + 1.7} \right)\left( {s + 1.8} \right)}}{{s\left( {s + 1.2} \right)\left( {s + 1.6} \right)}}.
\end{align*}
\begin{figure}[!ht]
\centering
\includegraphics[width=4cm]{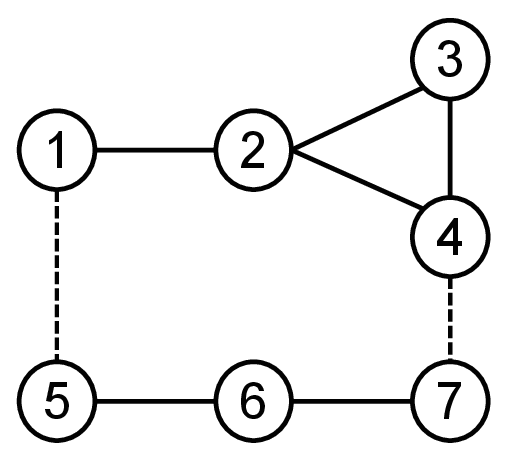}
\caption{The networks considered in the  Example.}
\label{fig: add a network}
\end{figure}
Suppose they are coupled by
$${\phi_{ij}(x)}=\begin{cases}
{a_{ij}\sin}( x ),&\mathrm{if } \left| {x} \right| < \frac{\pi }{2},\\
a_{ij}x,&\mathrm{otherwise},
\end{cases}\quad (i,j) \in \mathcal{E}_1\cup\mathcal{E}_2$$
where $a_{12}= 0.40$, $a_{23}= 0.32$, $a_{24}= 0.30$, $a_{34}= 0.35$, $a_{56}= 0.60$ and $a_{67}= 0.55$. We can obtain that ${{\overline{\alpha} }_{ij}} = a_{ij}$ for all $(i,j) \in \mathcal{E}_1\cup\mathcal{E}_2$. It can be obtained by solving the LMI in \cite[Lemma 2]{kottenstette2014relationships} that systems $H_1,\dots,H_7$ are IFP with indices $\nu_1 = -0.45$, $\nu_2 = -0.60$, $\nu_3 = -0.63$, $\nu_4=-0.65$, $\nu_5 = -0.40$, $\nu_6 =-0.54$ and $\nu_7=-0.51$, respectively. Additionally, we observe that the number of neighbours of system $H_i$ in $\mathcal{G}_1$ and $\mathcal{G}_2$ are $r_1=1$, $r_2=3$, $r_3=2$, $r_4=2$, $r_5=1$, $r_6=2$ and $r_7=1$. Note that $\frac{1}{{{\overline{\alpha} }_{ij}}} + {\nu _i} + {\nu _j} - \left( {{r_i} - 1} \right)\left| {{\nu _i}} \right| - \left( {{r_j} - 1} \right)\left| {{\nu _j}} \right| > 0$ for all $(i,j) \in \mathcal{E}_1\cup\mathcal{E}_2$. By Theorem~\ref{thm: network to network}, the networks $\mathcal{G}_1$ and $\mathcal{G}_2$ achieve IO consensus.\par 
In this example, the subnetworks $\mathcal{G}_1$ and $\mathcal{G}_2$ are interconnected via two coupling links, one connecting $H_1$ with $H_5$ and the other connecting $H_4$ with $H_7$. Suppose they are coupled by $${\phi _{ij}(x)}=\begin{cases}
{a_{ij}\sin( x )},&\text{if } \left| {x} \right| < \frac{\pi }{2},\\
a_{ij} x,&\mathrm{otherwise},
\end{cases}\, (i,j) \in \{(1,5),(4,7)\}$$ 
with $a_{15}=0.37$ and $a_{47}=0.16$. Setting $$\gamma_{15}=\min\limits_{(i,j)\in\{(1,2),(5,6)\}} \frac{{\frac{1}{{{\alpha _{ij}}}} - {r_i}\left| {{\nu _i}} \right| - {r_j}\left| {{\nu _j}} \right|}}{{\left| {{\nu _i}} \right|}} = 0.4667,$$ we obtain ${\gamma_{15}}\left( {\frac{1}{{{{\overline{\alpha}}_{15}}}} + {\nu _1} + {\nu _5}} \right) - r_1\left| {{\nu _1}} \right|- r_5\left| {{\nu _5}} \right|= 0.0147 > 0.$ Similarly, setting $$\gamma_{47}=\min\limits_{(i,j)\in\{(4,2),(4,3),(7,6)\}} \frac{{\frac{1}{{{\alpha _{ij}}}} - {r_i}\left| {{\nu _i}} \right| - {r_j}\left| {{\nu _j}} \right|}}{{\left| {{\nu _i}} \right|}} = 0.3589,$$ we have ${\gamma_{47}}\left( {\frac{1}{{{{\overline{\alpha}}_{47}}}} + {\nu _4} + {\nu _7}} \right) - r_4\left| {{\nu _4}} \right|- r_7\left| {{\nu _7}} \right|= 0.0168 > 0.$ By Theorem \ref{thm: network to network}, we can conclude that IO consensus is preserved in the augmented network. We simulate this by setting the initial value $Y(0) ={\left[ {\begin{array}{*{20}{c}}
-0.25 & -0.55 & -1.25 & -0.4 & -0.0875 &0.2 &1.36
\end{array}} \right]^{\top}} $ and the external signals $w_i(t) = 0.5{\bar w_i}(t)$, where ${\bar w_i}(t)$ is white Gaussian noise with ${\bar w_i}(t) \sim \mathcal{N}(0,1)$. Suppose the interconnection between the two networks occurs at $t=15s$. As shown in Fig. \ref{fig: example2}, the output trajectories within each subnetwork, $\mathcal{G}_1$ and $\mathcal{G}_2$, initially achieve consensus, and after the interconnection is established, the output trajectories of all systems reach a common consensus approximately as per Definition \ref{def: IO consensus}.
\begin{figure}[!ht]
\centering
\includegraphics[width=7cm]{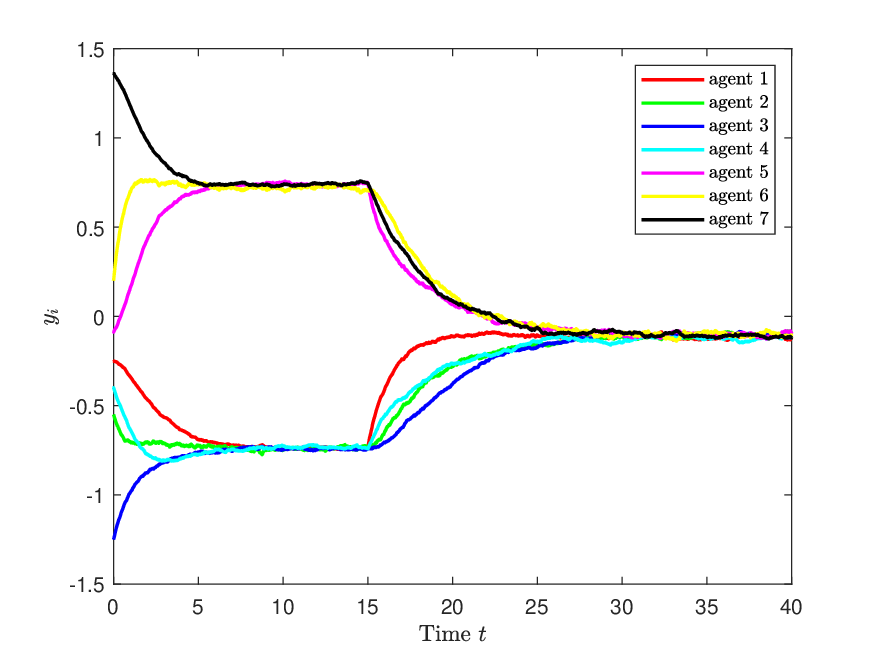}
\caption{Output trajectories of the systems in the Example.}
\label{fig: example2}
\end{figure}

\section{Conclusion}
This paper proposed a plug-and-play framework for output consensus in heterogeneous networks of IFP systems interconnected via nonlinear diffusive couplings. A passivity-compensation mechanism was developed, exploiting the surplus passivity in coupling links to locally compensate the shortages of passivity at the nodes. The analysis was conducted under two settings of the plug-and-play framework: the standard scenario in which a single system joins an existing network, and the more general scenario involving the interconnection of subnetworks. Locally verifiable interface conditions, expressed in terms of passivity indices and coupling gains, were derived to guarantee that consensus properties are preserved upon interconnection.

Future work will investigate more general forms of network plug-and-play consensus, including relaxing the constraints on interconnection edges and extending the framework to multiple interconnected networks.

\section*{Appendix}\label{sec: Appendix}
\noindent \emph{Proof of Lemma \ref{lem: positive definite}}. Recalling the definition of incidence matrix $D=[d_{ik}]\in\mathbb{R}^{n\times p}$, we have
\begin{equation*}
{d_{ik}} = \left\{ 
\begin{matrix}
     + 1, & k \in \mathscr{L}_i^ + \\
 - 1, & k \in \mathscr{L}_i^ - \\
0,& \mathrm{otherwise}.
\end{matrix}
\right.
\end{equation*}
Consequently, the matrix ${\Theta D} =[\varsigma_{ki}]\in\mathbb{R}^{p \times n}$, where
\begin{align*}
{\varsigma_{ik}} = \left\{ 
\begin{matrix}
     + \theta_i, & k \in \mathscr{L}_i^ + \\
 - \theta_i, & k \in \mathscr{L}_i^ - \\
0,& \mathrm{otherwise}.
\end{matrix}
\right.
\end{align*}
It can be observed that ${{D^{\top}}\Theta D} =[\zeta_{kl}]\in\mathbb{R}^{p \times p}$, where
\begin{align}\label{eq: defDEDT}
{\zeta_{kl}} = \left\{
\begin{matrix}
   \theta_i+\theta_j, &  k=l \in \mathscr{L}_i^ + \cap \mathscr{L}_j^ - \\
     \theta_i, & k \in \mathscr{L}_i^ + \cap \mathscr{L}_j^ -, l \in \mathscr{L}_i^ + \\
 - \theta_i, & k \in \mathscr{L}_i^ + \cap \mathscr{L}_j^ -, l \in \mathscr{L}_i^ - \\
    -  \theta_j, & k \in \mathscr{L}_i^ + \cap \mathscr{L}_j^ -, l \in \mathscr{L}_j^ + \\
  \theta_j, & k \in \mathscr{L}_i^ + \cap \mathscr{L}_j^ -, l \in \mathscr{L}_j^ - \\
0,& \mathrm{otherwise}.
\end{matrix}
\right.
\end{align}
Therefore, it follows that $M = [m_{kl}]\in\mathbb{R}^{p \times p}$, where $m_{kk} = {{\theta _i} + {\theta _j} + {\sigma_k}}$ with $k \in \mathscr{L}_i^ + \cap \mathscr{L}_j^ -$, and $\left| {{m_{kl}}} \right| = \left| {{\theta _i}} \right|$ if $l\in \mathscr{L}_i^ \pm$ and $\left| {{m_{kl}}} \right| = \left| {{\theta _j}} \right|$ if $l\in \mathscr{L}_j^ \pm$ for $l\neq k$.\par
Consider a diagonal matrix $S=\mathrm{diag}\left\{ {{s _1}, \ldots, s_p} \right\}\succ 0$ such that $SMS = [\bar m_{kl}] \in\mathbb{R}^{p\times p}$, where each entry satisfies $\bar m_{kl} = m_{kl}s_ks_l$. According to the Gershgorin circle theorem  \cite[p344]{horn2012matrix}, every eigenvalue of $SMS$ lies within the union of $p$ discs, each centered at $\bar m_{kk}$ of radius $\sum\limits_{l = 1,l \ne k}^p {\left| {{{\bar m}_{kl}}} \right|}$. Consequently, when $\bar m_{kk}>\sum\limits_{j = 1,j \ne i}^p {\left| {{{\bar m}_{kl}}} \right|}$ for all $k \in \{1, \ldots, p\}$, i.e., ${s_k}\left( {{\theta _i} + {\theta _j} + {\sigma_k}} \right) \ge \sum\limits_{{l\neq k},\,{l \in \mathscr{L}_i^ \pm }} {{s_l}\left| {{\theta _i}} \right|}  + \sum\limits_{{l\neq k},\,{l \in \mathscr{L}_j^ \pm }} {{s_l}\left| {{\theta _j}} \right|}$ with $k \in \mathscr{L}_i^ + \cap \mathscr{L}_j^ -$ for all $\left( {i,j} \right)\in \mathcal{E}$, all these discs lie in the closed right half-plane, ensuring that $SMS$ is positive, which implies that $M$ is positive.$\hfill\blacksquare$

\bibliographystyle{IEEEtran}
\bibliography{reference}
\end{document}